\RequirePackage{iftex}
\ifXeTeX
  \RequirePackage{scrlfile}
  \PreventPackageFromLoading{breakurl,hypdvips}
\fi
\documentclass[final,hidelinks]{siamart251216}
\usepackage{amsfonts,amssymb,mathrsfs,cases,bm,microtype}
\usepackage{etoolbox}
\makeatletter
\if@review
  \let\siamlinenomathAMS\linenomathWithnumbers
  \patchcmd{\siamlinenomathAMS}
    {\advance\postdisplaypenalty\linenopenalty}{}{}
    {\PackageError{manuscript}{Cannot patch math line numbering}{Check the lineno package version.}}
  \newcommand*\siamfixmathlinenumbers[1]{%
    \expandafter\patchcmd\csname #1\endcsname
      {\linenomath}{\siamlinenomathAMS}{}
      {\PackageError{manuscript}{Cannot patch #1 line numbering}{Check the SIAM class version.}}}
  \forcsvlist{\siamfixmathlinenumbers}
    {align,align*,alignat,alignat*,flalign,flalign*,gather,gather*,multline,multline*}
\fi
\makeatother
\newcommand{\Q}{\mathbb{Q}}
\newcommand{\tr}{\operatorname{Tr}}

\newcommand{\gf}{\mathbb{F}}

\newcommand{\D}{\mathcal{D}}
\newcommand{\cH}{\mathcal{H}}
\newcommand{\cQ}{\mathcal{Q}}
\newcommand{\cW}{\mathcal{W}}

\newsiamthm{conj}{Conjecture}
\newsiamremark{example}{Example}
\newsiamremark{rmk}{Remark}

\headers{Correlation Distribution of Niho Decimation}{Maosheng Xiong and Haode Yan}
\title{On Correlation Distribution of Niho Decimation \lowercase{$d=4(2^m-1)+1$}}
\author{Maosheng Xiong\thanks{Department of Mathematics, The Hong Kong
University of Science and Technology, Hong Kong
(\email{mamsxiong@ust.hk}).}
\and Haode Yan\thanks{Corresponding author. School of Science,
Harbin Institute of Technology, Shenzhen, 518055, PR China
(\email{yanhd@hit.edu.cn}).}}
\hypersetup{
  pdftitle={On Correlation Distribution of Niho Decimation},
  pdfauthor={Maosheng Xiong and Haode Yan}
}

\begin{document}
\maketitle

\begin{abstract}
Let $m$ be a positive integer, $q=2^m$, and $d=4(q-1)+1$.
We determine the complete cross-correlation distribution between
a binary $m$-sequence of period $q^2-1$ and its $d$-decimation.
By Niho's reduction, this distribution is determined by the number
of distinct roots of
$f_a(X)=X^7+aX^4+a^qX^3+1$
lying in the unit circle
$U_{q+1}=\{u\in\gf_{q^2}^*:u^{q+1}=1\}$,
as $a$ ranges over $\gf_{q^2}^*$. Combining the recent resolution of Niho's last conjecture with moment
identities reduces the problem to counting those coefficients $a$
for which $f_a$ has exactly five distinct roots in $U_{q+1}$--a task that appears intractable by direct methods. 

Our main innovation  
is to count instead the four-element subsets of
$U_{q+1}$ that lie in the root set of some $f_a$, exploiting the fact that each such subset
determines $a$ uniquely. Via normalization and cubic resolvents, we transform this problem into counting
solutions of certain equations over $\gf_q$. The resulting count is evaluated
using Kloosterman sums together with Carlitz's recurrence relation. As a result, we obtain explicit formulas for all correlation frequencies.
\end{abstract}



\section{Introduction and main result}
% ===== Begin included source: intro.tex =====
% Introduction and main result for Niho-XY-1.tex.

Let $\{s(t)\}_{t=0}^{2^n-2}$ be a binary $m$-sequence of period $2^n-1$. 
For an integer $d$, the $d$-decimation of $\{s(t)\}_{t=0}^{2^n-2}$ is the sequence
$\{s(dt)\}_{t=0}^{2^n-2}$, where the indices $dt$ are taken modulo $2^n-1$.
If $\gcd(d,2^n-1)=1$, then $\{s(dt)\}_{t=0}^{2^n-2}$ is also a binary $m$-sequence
of the same period.  The cross-correlation function between these
two sequences is
\begin{equation}\label{eq:correlation-definition}
 C_d(\tau)=\sum_{t=0}^{2^n-2}(-1)^{s(t+\tau)-s(dt)},
 \qquad 0\leq\tau\leq 2^n-2.
\end{equation}
In the theory of sequence design, it is interesting to find new 
decimations $d$ that give low cross-correlation and to 
determine the values $C_d(\tau)$ together with the frequency
of each value. These are classical problems in sequence design~\cite{SP80}.
The values and their frequencies constitute the correlation
distribution for the decimation $d$.  
Motivated by important applications, these problems have received
much attention since the 1960s, and many theoretical results have
been obtained~\cite{G68,GG05,H76,H78,HK98,N72,R04,T70,XLZH16,XLZH17,XY24,ZLFG14}.

In his influential Ph.D. thesis in 1972, Niho~\cite{N72} studied
the cross-correlation of binary $m$-sequences of period
$2^{2m}-1$ whose decimation has the form 
\begin{equation*}
 d=s(2^m-1)+1,\qquad \gcd(d,2^{2m}-1)=1. 
\end{equation*}
He determined the cross-correlation distribution for many such
decimations and proposed a series of conjectures and open
problems that have inspired research for the past 
five decades.  For example, later work produced further families such as those with four-valued cross-correlation
\cite{D98,DFHR06,HR05} and established restrictions on the possible
correlation values~\cite{C04,HLR07}.  The study of these decimations
also connects sequence design with the Walsh spectra of power
functions and the weight distributions of cyclic
codes~\cite{C04,HKL21,K12,LFG14}. In recognition of his work,
these decimations are called Niho-type decimations. 
Niho's work was subsequently extended to odd
characteristic~\cite{R04}.  For further
applications and open problems, we refer to the survey~\cite{LZ19}
and the references therein.

Niho~\cite{N72} determined the distribution for $s=2$, the
simplest nontrivial case.  For $s=3$, Xia et al.~\cite{XLZH16}
determined the distribution for all admissible $m\geq3$,
completing earlier work of Dobbertin et al.~\cite{DFHR06}.
The corresponding distributions in odd characteristic were
determined for $p=3$ by Xia et al.~\cite{XLZH17} and for primes
$p\geq5$ by Xiong and Yan~\cite{XY24}, under the decimation
condition $\gcd(d,p^{2m}-1)=1$.

For $s=4$, the last conjecture in Niho's Ph.D. thesis asserted
that the correlation takes at most five values when $m$
is even.  Helleseth, Katz, and Li~\cite[Theorem~3]{HKL21}
resolved this conjecture and also showed that at most six values
occur when $m$ is odd.  The remaining problem is to determine
the complete correlation distribution for $s=4$.  This was posed
as Open Problems~7 and~9 in~\cite{HL24} for even and odd $m$,
respectively. 

In this paper, we resolve both open problems by determining the
complete correlation distribution for $d=4(2^m-1)+1$.
This decimation satisfies $\gcd(d,2^{2m}-1)=1$ for every positive
integer $m$. The following theorem gives the distribution in terms
of the integer sequence $(T_k)_{k\ge0}$, defined by
\begin{equation}\label{1:tm}
T_k=-\left(\frac{1+\sqrt{-15}}{2}\right)^k
    -\left(\frac{1-\sqrt{-15}}{2}\right)^k.
\end{equation}
Equivalently, $T_0=-2$, $T_1=-1$, and
\[
T_{k+2}-T_{k+1}+4T_k=0\qquad(k\ge0).
\]

\begin{theorem}\label{distribution}
Let $m$ be a positive integer, $q=2^m$, and $d=4(q-1)+1$.
The cross-correlation distribution of a binary $m$-sequence of
period $q^2-1$ and its $d$-decimation is
\[
\begin{array}{c|c}
\hline
\rule{0pt}{2.6ex}\textnormal{Value of }C_d(\tau)&\textnormal{Frequency}\\[1mm]
\hline
\rule[-2.6ex]{0pt}{6.2ex}-q-1&\displaystyle
 \frac{11q^2-17q+(-1)^m(-T_m+5q-1)}{30}\\
\hline 
\rule[-2.6ex]{0pt}{6.2ex}-1&\displaystyle
 \frac{9q^2+5q-22+(-1)^m(3T_m-11q+1)}{24}\\
\hline
\rule[-2.6ex]{0pt}{6.2ex}q-1&\displaystyle
 \frac{q^2+2q-2+(-1)^m(-T_m+2q+1)}{6}\\
\hline
\rule[-2.6ex]{0pt}{6.2ex}2q-1&\displaystyle
 \frac{q^2-q+6+(-1)^m(T_m+q-5)}{12}\\
\hline
\rule[-2.6ex]{0pt}{6.2ex}3q-1&\displaystyle
 \frac{(q-2)(1-(-1)^m)}{6}\\
\hline
\rule[-2.6ex]{0pt}{6.2ex}4q-1&\displaystyle
 \frac{q^2-7q+10-(-1)^m(T_m-5q+11)}{120}.\\
\hline
\end{array}
\]
\end{theorem}

For example, $m=4$ and $m=5$ give the following correlation 
distributions:
\begingroup
\renewcommand{\arraystretch}{1.4}
\[
\begin{aligned}
m=4:\quad &
\begin{array}{c|rrrrr}
\hline
\textnormal{Value of } C_d(\tau)&-17&-1&15&31&63\\ \hline
\textnormal{Frequency}&88&89&56&20&2\\
\hline
\end{array} \\[2mm]
m=5:\quad &
\begin{array}{c|rrrrrr}
\hline
\textnormal{Value of } C_d(\tau)&-33&-1&31&63&95&127\\ \hline
\textnormal{Frequency}&350&412&160&86&10&5\\
\hline
\end{array}
\end{aligned}
\]
\endgroup
Here $T_4=-17$ and $T_5=-61$.

We briefly describe the proof.  Let $N_i$ be the number of
$a\in\gf_{q^2}$ for which
$f_a(X)=X^7+aX^4+a^qX^3+1$ has exactly $i$ distinct roots in the unit circle
$U_{q+1}=\{u\in\gf_{q^2}^*:u^{q+1}=1\}$.  Niho's reduction shows that the correlation distribution
is determined by the frequencies $N_i$. The known root restrictions~\cite{HKL21} and standard
moment identities~\cite{DFHR06} reduce the problem to determining
$N_5$.  For $s=3$, the analogous count follows directly from the number of
weight-five codewords in a binary Zetterberg code~\cite{XLZH16}.
Determining this count for $s=4$ is the main difficulty in the proof.

The key observation is that four distinct roots in $U_{q+1}$
determine $f_a$ uniquely.  By scaling these roots so that their
product is one, we reduce the calculation of $N_5$ to counting
certain four-element subsets of $U_{q+1}$
(Theorem~\ref{thm:WI}).  Cubic resolvents then turn this into
counting solutions of equations over $\gf_q$.  We count these
solutions using affine conics and character sums.  The resulting mixed
Kloosterman sum is evaluated using a formula of Carlitz~\cite{C69}.
The moment identities then give the
remaining frequencies and complete the distribution.

Section~\ref{preliminaries} gives the notation and basic facts about finite fields. Section~\ref{rootfreq} reduces the distribution problem to determining $N_5$.
Section~\ref{fourroot} converts this task into counting normalized quartets,
and Section~\ref{resolventsection} expresses the quartet count
in terms of the cardinalities of two solution sets over $\gf_q$.
Section~\ref{coniccounts} computes these cardinalities, and
Section~\ref{main} completes the distribution.
Appendices~\ref{app:determinant} and~\ref{Kloostermansum} contain the
determinant proof and the supporting Kloosterman identities.
% ===== End included source: intro.tex =====

\section{Preliminaries}\label{preliminaries}
% ===== Begin included source: pre.tex =====
In this section we collect some notation and elementary facts used below. 

Throughout this paper, let $m$ be a positive integer, $q=2^m$ and $d=4(q-1)+1=4q-3$.
Let $\gf_q$ denote the finite field with $q$ elements, and let
$\gf_{q^2}$ be its quadratic extension.
For any finite field $F$, write $F^*=F\setminus\{0\}$.
For $z\in\gf_{q^2}$, write $\bar z=z^q$.
The unit circle is
\[U_{q+1}=\{u\in\gf_{q^2}^*:u^{q+1}=1\}.\]
For $u\in U_{q+1}$, we have $\bar u=u^{-1}$, and
$U_{q+1}\cap\gf_q=\{1\}$.

For finite fields $F\subseteq E$, let
$\tr_{E/F}:E\to F$ denote the relative trace.
In particular, $\tr_{\gf_{q^2}/\gf_q}(z)=z+\bar z$
for $z\in\gf_{q^2}$.
We write $\operatorname{tr}=\tr_{\gf_q/\gf_2}$ for short, so
\[
\operatorname{tr}(x)=\tr_{\gf_q/\gf_2}(x)=\sum_{j=0}^{m-1}x^{2^j},
\qquad x\in\gf_q.
\]
It satisfies $\operatorname{tr}(x^2)=\operatorname{tr}(x)$
for $x\in\gf_q$.
Trace transitivity gives
$\tr_{\gf_{q^2}/\gf_2}(z)=\operatorname{tr}(z+\bar z)$
for $z\in\gf_{q^2}$.

The canonical additive characters of $\gf_q$ and $\gf_{q^2}$
are defined by
\[
\chi_q(x)=(-1)^{\operatorname{tr}(x)}
\quad(x\in\gf_q),\qquad
\chi_{q^2}(z)=(-1)^{\tr_{\gf_{q^2}/\gf_2}(z)}
\quad(z\in\gf_{q^2}).
\]
We write $\chi=\chi_q$ for short.
Thus $\chi_{q^2}(z)=\chi(z+\bar z)$ for $z\in\gf_{q^2}$,
and $\chi(1)=(-1)^m$.

The following criterion will be used many times in this paper. 

\begin{lemma}[\cite{BRS67}]\label{rootinfq}
Let $n$ be a positive integer and let $a,b\in\gf_{2^n}$ with $a\ne0$.
The equation $x^2+ax+b=0$ has two distinct solutions in $\gf_{2^n}$
if and only if $ \tr_{\gf_{2^n}/\gf_2}(b/a^2)=0$. 
\end{lemma}

Consequently, for $a\in \gf_q^*$ and $b\in \gf_q$,
\[
 \left|\left\{x\in \gf_q:x^2+ax+b=0\right\}\right|=1+\chi(b/a^2).
\]
In particular,
\[
 \left|\left\{x\in \gf_q:x^2+x+t=0\right\}\right|=1+\chi(t)\qquad \forall \, t\in \gf_q. 
\]
When $a=0$, the equation $x^2+b=0$ has a unique root in $\gf_q$.

We also use the orthogonality relation: 
\begin{eqnarray*} 
 \sum_{x\in \gf_q}\chi(tx)=
 \begin{cases}
 q,&t=0,\\
 0,&t\ne0,
 \end{cases}
 \qquad t\in \gf_q.
\end{eqnarray*}

\section{Root frequencies and moment reduction}\label{rootfreq}

For $m=1$, we have $q=2$ and $d=5$.
Direct computation gives the correlation values $3$ and $-1$
with frequencies $1$ and $2$, respectively, in agreement with
Theorem~\ref{distribution}.
In Sections~\ref{rootfreq}--\ref{main}, we therefore assume $m\ge2$.

To apply the moment method, we first express the correlation
distribution through root counts on $U_{q+1}$.
For $a\in\gf_{q^2}$, define
\begin{equation}\label{fa}
	f_a(X)=X^7+aX^4+\bar aX^3+1.
\end{equation}
Write $N(a)=|\{u\in U_{q+1}:f_a(u)=0\}|$.
For $0\leq i\leq7$, let
$N_i=|\{a\in\gf_{q^2}:N(a)=i\}|$.

Choose a primitive element $\xi$ of $\gf_{q^2}$.
Up to a cyclic shift, the binary $m$-sequence can be written as
\[
s(t)=\tr_{\gf_{q^2}/\gf_2}(\xi^t),
\qquad 0\leq t\leq q^2-2.
\]
Such a shift changes only the order of the cross-correlation values,
so it does not change their distribution.
The standard Niho reduction \cite[Lemma~8]{HKL21} gives
\begin{equation*}
	\cW(a):=\sum_{x\in\gf_{q^2}}
	\chi_{q^2}(x^d+ax)=q(N(a)-1).
\end{equation*}
For $a=\xi^\tau$, substituting $x=\xi^t$ in
\eqref{eq:correlation-definition} gives
\[
C_d(\tau)=\sum_{x\in\gf_{q^2}^*}
\chi_{q^2}(x^d+ax)=\cW(a)-1.
\]
As $\tau$ runs through the $q^2-1$ shifts,
$a$ runs through $\gf_{q^2}^*$.
Thus the correlation distribution of $C_d(\tau)$ is obtained
from the $N_i$'s after excluding $a=0$.

The frequencies $N_i$ satisfy the Niho power-moment
identities \cite[Lemma~5]{DFHR06}.

\begin{equation}\label{moment}
	\left\{
	\begin{aligned}
		\sum_{i=0}^7N_i&=q^2,\\
		\sum_{i=0}^7(i-1)N_i&=q,\\
		\sum_{i=0}^7(i-1)^2N_i&=q^2,\\
		\sum_{i=0}^7(i-1)^3N_i&=b_3q.
	\end{aligned}
	\right.
\end{equation}
Here $b_3$ counts the solutions $x\in\gf_{q^2}$ of
\begin{equation}\label{b=1}
	(x+1)^{d}+x^{d}=1.
\end{equation}
We first evaluate $b_3$, which determines the third moment
in \eqref{moment}.
\begin{lemma}\label{b3}
	The number $b_3$ of solutions of \eqref{b=1} is
	\begin{equation}\label{eq:b3-value}
		b_3=q+1-(-1)^m.
	\end{equation}
\end{lemma}

\begin{proof}
	Both $x=0$ and $x=1$ satisfy \eqref{b=1}.
	For $x\notin\{0,1\}$, multiply both sides of \eqref{b=1}
	by $x^3(x+1)^3$. The equation becomes
	\[
	(x+1)^{4q}x^3+x^{4q}(x+1)^3+x^3(x+1)^3=0,
	\]
	which factors as
	\[
	(x^q+x)^4(x^2+x+1)=0.
	\]
	The first factor vanishes exactly for $x\in\gf_q$.
	For the second factor, choose
	$\omega\in\gf_4\setminus\gf_2$, viewed as an element
	of $\gf_{q^2}$; its roots are $\omega,\omega^2$.
	Thus the full solution set is
	$\gf_q\cup\{\omega,\omega^2\}$.
	For even $m$, the roots $\omega,\omega^2$ lie in $\gf_q$,
	so $b_3=q$.
	For odd $m$, they lie outside $\gf_q$, so $b_3=q+2$.
	Both cases give \eqref{eq:b3-value}.
\end{proof}

To reduce the number of unknown frequencies, we use
the root restrictions in \cite[Theorem~9]{HKL21}:
\begin{equation}\label{3:n67}
	N_6=N_7=0,
\end{equation}
and
\begin{equation}\label{3:n4-even}
	N_4=0 \qquad\text{if }m\text{ is even}.
\end{equation}
The remaining case for $N_4$ follows from
\cite[Proposition~30]{HKL21};
we include a proof for completeness.

\begin{lemma}\label{N4}
	Let $m$ be odd. Then $N_4=(q-2)/3$.
\end{lemma}

\begin{proof}
	Proposition~30 of~\cite{HKL21} states that for any
	$a\in\gf_{q^2}$, if $f_a(X)$ is separable, then
	$N(a)\notin\{4,6,7\}$.
	Thus to compute $N_4$, it suffices to consider
	$a\in\gf_{q^2}^*$ such that $f_a$ is inseparable.
	Since $q$ is even, $f_a'(X)=X^2(X^4+\bar a)$.
	Since $0$ is not a root of $f_a$, a common root $r$
	of $f_a$ and $f_a'$ satisfies $r^4=\bar a\ne0$.
	Substitution into $f_a$ gives $f_a(r)=a\bar a+1$.
	Hence $f_a$ is inseparable precisely when $a\in U_{q+1}$,
	and in that case $f_a(X)=(X^4+a^{-1})(X^3+a)$.
	The first factor has a unique root of multiplicity four,
	which lies in $U_{q+1}$.
	The two factors have a common root only when $a^7=1$,
	which, because $\gcd(7,q+1)=1$, forces $a=1$.
	Thus if $f_a$ is inseparable, it has at most four
	distinct roots.
	
	Since $m$ is odd, the cubes form the subgroup
	$U_{q+1}^{(3)}=\{u^3:u\in U_{q+1}\}$ of order $(q+1)/3$.
	For $a\in U_{q+1}^{(3)}\setminus\{1\}$, all three roots
	of $X^3+a$ lie on $U_{q+1}$, so $N(a)=4$;
	for $a\in U_{q+1}\setminus U_{q+1}^{(3)}$, only the root
	of $X^4+a^{-1}$ lies on $U_{q+1}$; and $N(1)=3$.
	Consequently $N_4=(q+1)/3-1=(q-2)/3$.
\end{proof}

Combining \eqref{3:n4-even} with Lemma~\ref{N4} gives
$N_4=(1-(-1)^m)(q-2)/6$.
Using this value, $N_6=N_7=0$, and the value of $b_3$
from Lemma~\ref{b3}, we solve the moment equations \eqref{moment}
to express $N_0,N_1,N_2,N_3$ in terms of $N_5$:
\begin{equation}\label{eq:frequencies-in-N5}
	\left\{
	\begin{aligned}
		N_0&=\frac{q^2-q-1+(-1)^m}{3}+4N_5,\\
		N_1&=\frac{3q^2-4q+8+(-1)^m(q-8)}6-15N_5,\\
		N_2&=\frac{3q-4-(-1)^m(q-4)}2+20N_5,\\
		N_3&=\frac{q^2-4q+8+(-1)^m(3q-8)}6-10N_5.
	\end{aligned}
	\right.
\end{equation}
Together with the known values of $N_4,N_6,N_7$, these
formulas determine all root frequencies once $N_5$ is known.
Thus it remains to determine $N_5$.
In the next section, we reduce this problem to counting
normalized quartets in $U_{q+1}$.

\section{Four-root incidences}\label{fourroot}

To determine $N_5$, we consider the following four-root
incidence set
\begin{equation}\label{eq:I-definition}
 I=\left\{(a,S):a\in \gf_{q^2},\ S\subset U_{q+1},\ |S|=4,
       \ f_a(u)=0\text{ for every }u\in S\right\}.
\end{equation}
Counting first by $a$ gives
\begin{equation}\label{eq:I-root-count}
 |I|=\sum_{a\in \gf_{q^2}}\binom{N(a)}4=N_4+5N_5,
\end{equation}
where \eqref{3:n67} has been used in the last equality.  To compute its cardinality, we identify $I$ with a normalized family of four-subsets that can 
eventually be counted over the base field $\gf_q$.

For a four-element set $V=\{v_1,v_2,v_3,v_4\}\subset U_{q+1}$, let
$\sigma_j(V)$ denote its $j$th elementary symmetric function for $j=1,2,3,4$, and write $b=\sigma_1(V)$, $c=\sigma_2(V)$,
$h=\sigma_3(V)$, and $p=\sigma_4(V)$, and set $\Delta(V)=\prod_{1\leq i<j\leq4}(v_i+v_j)$. Since $v_i^{-1}=\bar v_i$ for each $i$, we have
\begin{eqnarray} \label{3:bcp}
c=p\left(\sum_{1\leq i<j\leq 4} \frac{1}{v_iv_j}\right)=p\bar{c}, \quad h=p\left(\sum_{i=1}^4\frac{1}{v_i}\right)=p\bar{b}.
\end{eqnarray}
Define
\begin{equation*}
 D(V)=\det\begin{pmatrix}
 v_1^7&v_1^4&v_1^3&1\\
 v_2^7&v_2^4&v_2^3&1\\
 v_3^7&v_3^4&v_3^3&1\\
 v_4^7&v_4^4&v_4^3&1
 \end{pmatrix}.
\end{equation*}
An elementary proof of the following determinant identity
using column operations is given in Appendix~\ref{app:determinant}.
\begin{lemma}\label{lem:determinant}
For a four-element set
$V=\{v_1,v_2,v_3,v_4\}\subset U_{q+1}$, we have
\begin{equation}\label{eq:correct-determinant}
 D(V)=\Delta(V)p^2
 \bigl((b^2+c)(\bar b^2+\bar c)+b\bar b\bigr).
\end{equation}
Consequently $D(V)=0$ if and only if $(b^2+c)(\bar b^2+\bar c)=b\bar b$.
\end{lemma}

Define
\begin{eqnarray} \label{3:ZW1}
 Z=\left\{V\subset U_{q+1}:|V|=4,\ D(V)=0\right\},
\qquad
 W=\left\{V\in Z:\sigma_4(V)=1\right\}.
\end{eqnarray}

For a four-element set $V=\{v_1,v_2,v_3,v_4\}\subset U_{q+1}$ and $\lambda \in U_{q+1}$, denote $V^{-1}:=\left\{v_1^{-1},v_2^{-1},v_3^{-1},v_4^{-1}\right\}$ and $\lambda V:=\left\{\lambda v_1,\lambda v_2,\lambda v_3,\lambda v_4\right\}$. It is easy to see that if $V\in Z$, then $\lambda V\in Z$
for every $\lambda\in U_{q+1}$. Obviously $\sigma_4(V) \in U_{q+1}$ and
$\sigma_4(\lambda V)=\lambda^4\sigma_4(V)$.
Since the map $\lambda\mapsto\lambda^4$ permutes $U_{q+1}$,
there is a unique $\lambda\in U_{q+1}$ such that
$\lambda V\in W$. This implies that
\begin{eqnarray*} 
|W|=\frac{|Z|}{q+1}.
\end{eqnarray*}
We will prove that there is a bijection $I \longleftrightarrow W$ so $|I|=|W|$. To prove this main result, we need some lemmas. 

\begin{lemma}\label{lem:sparse-rank}
	A nonzero polynomial $AX^4+BX^3+C\in\gf_{q^2}[X]$
	cannot have four distinct roots in $U_{q+1}$.
\end{lemma}

\begin{proof}
	If $A=0$, the polynomial has degree at most three and cannot
	have four distinct roots. Otherwise, dividing the polynomial by $A$
	does not change its roots. Thus it suffices to consider a polynomial
	of the form $X^4+BX^3+C$.
	Suppose that the polynomial $X^4+BX^3+C$ has four distinct roots
	$v_1,v_2,v_3,v_4\in U_{q+1}$, and put
	$V=\{v_1,v_2,v_3,v_4\}$. Then
	
\[X^4+BX^3+C=\prod_{i=1}^4\left(X+v_i\right)=X^4+\sigma_1(V)X^3+\sigma_2(V)X^2+\sigma_3(V)X+\sigma_4(V). \]
Its $X$-coefficient is zero, so $\sigma_3(V)=0$.  But by \eqref{3:bcp}, $\sigma_3(V)=\sigma_4(V)\overline{\sigma_1(V)}$,
and hence $\sigma_1(V)=0$, that is, $B=0$. Hence the polynomial is $X^4+C$. This polynomial has a unique root in $\gf_{q^2}$,
a contradiction.
\end{proof}
For $V=\{v_1,v_2,v_3,v_4\}\in Z$, the condition $D(V)=0$
implies that there is a nonzero polynomial $F(X)=\alpha X^7+\beta X^4+\gamma X^3+\delta$, $ \alpha,\beta,\gamma,\delta\in\gf_{q^2}$
such that $F(v_i)=0$ for $i=1,2,3,4$.
By Lemma~\ref{lem:sparse-rank}, $\alpha\ne0$.
Dividing by $\alpha$ gives a monic polynomial, which is unique
by the same lemma. We denote it by
\begin{eqnarray} \label{3:pv} 
 P_V(X)=X^7+A_VX^4+B_VX^3+C_V.
\end{eqnarray}
Conjugate reciprocity gives the following relations
among its coefficients.

\begin{lemma}\label{lem:sparse-polynomial}
For $V\in Z$, the polynomial $P_V$ defined in \eqref{3:pv}
has the form
\begin{equation}\label{eq:sparse-polynomial}
 P_V(X)=X^7+A_VX^4+C_V\bar A_VX^3+C_V,
 \qquad C_V\in U_{q+1}.
\end{equation}
Moreover, for $\mu\in U_{q+1}$,
\begin{equation}\label{eq:sparse-scaling}
 P_{\mu V}(X)=\mu^7P_V(X/\mu),\qquad
 A_{\mu V}=\mu^3A_V,\qquad C_{\mu V}=\mu^7C_V.
\end{equation}
\end{lemma}

\begin{proof}
Put
\[
 P_V^\dagger(X)=X^7\overline{P_V}(X^{-1})
 =\bar C_VX^7+\bar B_VX^4+\bar A_VX^3+1,
\]
where $\overline{P_V}(X)=X^7+\bar A_VX^4+\bar B_V X^3+\bar C_V$. 
For every $v\in V$, we have $v^q=v^{-1}$, and hence
\[P_V^\dagger(v)=v^7 \overline{P_V}(v^{-1})=v^7 \overline{P_V(v)}=0, \]
that is, $P_V^\dagger (X)$ vanishes on $V$.  If $C_V=0$, then
$P_V^\dagger(X)=\bar B_VX^4+\bar A_VX^3+1$ vanishes on $V$,
contradicting Lemma~\ref{lem:sparse-rank}.  Thus $C_V\ne0$.
After dividing by
$\bar C_V$, uniqueness of the monic polynomial gives
$P_V^\dagger=\bar C_VP_V$.  Comparing coefficients yields
$C_V\bar C_V=1$ and $B_V=C_V\bar A_V$, proving
\eqref{eq:sparse-polynomial}.  Formula \eqref{eq:sparse-scaling}
follows from the same uniqueness after scaling the roots.
\end{proof}

\begin{theorem} \label{thm:WI}
There is a bijection between $W$ and $I$, where $W$ and $I$ are defined
in \eqref{3:ZW1} and \eqref{eq:I-definition}, respectively. In particular,
\begin{equation}\label{eq:W-incidence}
 |W|=|I|=N_4+5N_5.
\end{equation}
\end{theorem}

\begin{proof}
Let $V\in W$, and let $P_V(X)$ be the associated polynomial given in \eqref{3:pv}.  Since $C_V\in U_{q+1}$ and $\gcd(7,q+1)=1$,
there is a unique $\mu\in U_{q+1}$ such that $\mu^7C_V=1$.  By
\eqref{eq:sparse-scaling},
\[
 P_{\mu V}(X)=\mu^7P_V(X/\mu)
 =X^7+aX^4+\bar aX^3+1,
 \qquad a=\mu^3A_V.
\]
Thus
\[
 \Phi:W\longrightarrow I,\qquad
 \Phi(V)=(\mu^3A_V,\mu V)
\]
is well defined.

Conversely, take $(a,S)\in I$ and put $p=\sigma_4(S)$.  Choose the
unique $\lambda\in U_{q+1}$ with $\lambda^4p=1$, and set $V=\lambda S=\left\{v_1,v_2,v_3,v_4\right\}$.
The polynomial
\[
 \lambda^7f_a(X/\lambda)
 =X^7+\lambda^3aX^4+\lambda^4\bar aX^3+\lambda^7
\]
vanishes on $V$.  Hence 
\[
\begin{pmatrix}
 v_1^7&v_1^4&v_1^3&1\\
 v_2^7&v_2^4&v_2^3&1\\
 v_3^7&v_3^4&v_3^3&1\\
 v_4^7&v_4^4&v_4^3&1
 \end{pmatrix} 
\begin{pmatrix}
 1\\
 \lambda^3a\\
\lambda^4\bar{a}\\
\lambda^7 
\end{pmatrix}= \bm 0, 
\]
so we have $D(V)=0$, while $\sigma_4(V)=1$, so $V\in W$.

By uniqueness, $\lambda^7f_a(X/\lambda)=P_V(X)$;
consequently, $A_V=\lambda^3a$ and $C_V=\lambda^7$. By the choice of $\mu$, we have
$\mu^7C_V=1$, and therefore $(\mu\lambda)^7=1$.  Since seventh
powering is injective on $U_{q+1}$, $\mu=\lambda^{-1}$, and $ \Phi(V)=(\mu^3A_V,\mu V)=(a,S)$.
This also shows that the inverse is $(a,S)\longmapsto\sigma_4(S)^{-1/4}S$.
The counting identity now follows from \eqref{eq:I-root-count}.
\end{proof} 
% ===== End included source: pre.tex =====

\section{Cubic resolvents and trace conditions}\label{resolventsection}
% ===== Begin included source: chap3.tex =====
In this section, we use cubic resolvents to reduce the
computation of $|W|$, with $W$ defined in \eqref{3:ZW1},
to counting ordered pairs and triples over $\gf_q$
satisfying equations and trace-one conditions.

For $\varepsilon\in\{0,1\}$, define
\[\D_{\varepsilon}=\left\{r\in\gf_q^*:\operatorname{tr}(r^{-1})=\varepsilon\right\}.\]
Let
\[
\cQ=\{V\subset U_{q+1}:|V|=4,\ \sigma_4(V)=1\}.
\]
Obviously $W \subseteq \cQ$. For
$V=\{v_1,v_2,v_3,v_4\}\in\cQ$, define its cubic resolvent set by
\begin{equation*}
 \operatorname{Res}(V)=\{r_1,r_2,r_3\},
\end{equation*}
where $r_1=v_1v_2+v_3v_4$, $r_2=v_1v_3+v_2v_4$ and $r_3=v_1v_4+v_2v_3$.
The three entries are distinct; for instance,
\[r_1+r_2=v_1v_2+v_3v_4+v_1v_3+v_2v_4=(v_1+v_4)(v_2+v_3)\ne0.\]  
Since $\sigma_4(V)=v_1v_2v_3v_4=1$, every $r_i$ is of the form $u+u^{-1}$ for some $u\in U_{q+1}$. Thus $r_i \in \gf_q$ for each $i$. If $u=1$ then $r_i=u+u^{-1}=0$; if $u \in U_{q+1}\setminus \{1\}$, with $r_i=u+u^{-1}$, we see that $u,u^{-1}$ are the two distinct roots of the polynomial $X^2+r_iX+1$, and since $ u,u^{-1} \not \in \gf_q$, by Lemma~\ref{rootinfq}, we have $\operatorname{tr}(r_i^{-2})=\operatorname{tr}(r_i^{-1})=1$.
Consequently 
\begin{equation*}
 \operatorname{Res}(V)\subset\D_1\cup\{0\}.
\end{equation*}
%Conversely, every element of $\D_1\cup\{0\}$ is $u+u^{-1}$ for an element $u\in U_{q+1}$: take $u=1$ for zero, while for nonzero $r \in \D_1$ the roots of $X^2+rX+1$ are reciprocal elements of $U_{q+1}$.

The following lemma relates $r_1,r_2,r_3$ to $b$ and $c$.

\begin{lemma}\label{expressv}
Let $V\in\cQ$, put $b=\sigma_1(V)$ and $c=\sigma_2(V)$, and write
$\operatorname{Res}(V)=\{r_1,r_2,r_3\}$.
The monic cubic polynomial with roots $r_1,r_2,r_3$ is given by
\begin{equation}\label{eq:resolvent-cubic}
 \prod_{r\in\operatorname{Res}(V)}(T+r)
 =T^3+cT^2+b\bar bT+(b+\bar b)^2.
\end{equation}
Equivalently,
\begin{align}
 c&=r_1+r_2+r_3,\notag\\
 b\bar b&=r_1r_2+r_1r_3+r_2r_3,\notag\\
 b+\bar b&=\sqrt{r_1r_2r_3}.\label{eq:res-coeff-3}
\end{align}
\end{lemma}

\begin{proof}
	For four variables in characteristic two, direct computation gives
	\[
	\begin{aligned}
		r_1+r_2+r_3&=\sigma_2,\\
		r_1r_2+r_1r_3+r_2r_3&=\sigma_1\sigma_3,\\
		r_1r_2r_3&=\sigma_1^2\sigma_4+\sigma_3^2.
	\end{aligned}
	\]
For $V\in\cQ$, we have $\sigma_4(V)=1$ and hence
$\sigma_3(V)=\bar b$ by \eqref{3:bcp}. Substitution proves \eqref{eq:resolvent-cubic}; uniqueness of square
roots in characteristic two gives \eqref{eq:res-coeff-3}. 
\end{proof}

\begin{theorem} \label{resv}
The map
\[
 \operatorname{Res}:\cQ\longrightarrow
 \{R\subset\D_1\cup\{0\}:|R|=3\}
\]
is surjective, and
\begin{equation*}
 \operatorname{Res}^{-1}(\operatorname{Res}(V))=\left\{V,V^{-1}\right\},
\end{equation*}
where $\operatorname{Res}^{-1}(R)
=\{V\in\cQ:\operatorname{Res}(V)=R\}$
and $V^{-1}=\{v^{-1}:v\in V\}$. Moreover,
\[
 V=V^{-1}\quad\Longleftrightarrow\quad
 0\in\operatorname{Res}(V).
\]
\end{theorem}

\begin{proof}
Let $R=\{r_1,r_2,r_3\}\subset\D_1\cup\{0\}$ with $|R|=3$. For each $i\in\{1,2,3\}$, choose $u_i\in U_{q+1}$
such that $u_i+u_i^{-1}=r_i$. Such $u_i$ exist since $u_i=1$ works for $r_i=0$, while for
$r_i\in\D_1$, the polynomial $X^2+r_iX+1$ is irreducible
over $\gf_q$ and its roots in $\gf_{q^2}$ are conjugate
with product one, hence lie in $U_{q+1}$. Put
$\gamma=\sqrt{u_1u_2u_3}\in U_{q+1}$, and define
\begin{equation*}
 V=\{\gamma,u_1/\gamma,u_2/\gamma,u_3/\gamma\}.
\end{equation*}
It is easy to check that the four entries of $V$ are all distinct and $\sigma_4(V)=1$. Thus $V \in \cQ$. Moreover, for $\{i,j,k\}=\{1,2,3\}$, pairing $\gamma$
with $u_i/\gamma$ gives
\[
 \gamma \frac{u_i}{\gamma}+\frac{u_j}{\gamma}\frac{u_k}{\gamma}=u_i+\frac{u_ju_k}{\gamma^2}=u_i+u_i^{-1}=r_i,
\]
so $\operatorname{Res}(V)=R$. This proves that the map $\operatorname{Res}$ is surjective. 

To determine all quartets $V\in\cQ$ with
$\operatorname{Res}(V)=R$, set
\[
\rho_1=r_1+r_2+r_3,\qquad
\rho_2=r_1r_2+r_1r_3+r_2r_3,\qquad
\rho_3=r_1r_2r_3.
\]
For any such $V$, put $b=\sigma_1(V)$ and $c=\sigma_2(V)$.
By Lemma~\ref{expressv}, we have $c=\rho_1$,
$b\bar b=\rho_2$, and $b+\bar b=\sqrt{\rho_3}$.
Hence
\begin{equation}\label{eq:b-quadratic}
	b^2+\sqrt{\rho_3}\,b+\rho_2=0.
\end{equation}
The value $c=\rho_1$ is determined by $R$.
Once $b$ is fixed, $\bar b=b+\sqrt{\rho_3}$ is also determined,
and hence so is the polynomial
\begin{equation}\label{3:quartetp}
	\prod_{v\in V}(X+v)=X^4+bX^3+cX^2+\bar bX+1.
\end{equation}
Since $V$ is exactly the set of roots of this polynomial,
each possible value of $b$ determines at most one quartet $V$.
Equation~\eqref{eq:b-quadratic} has at most two roots,
so $R$ has at most two preimages.

For any such $V$, the set $V^{-1}$ also belongs to $\cQ$.
Since $v^{-1}=v^q$ for $v\in U_{q+1}$ and
$r_i\in\gf_q$ for $i=1,2,3$, we have
$\operatorname{Res}(V^{-1})=\{r_1^q,r_2^q,r_3^q\}=R$
and $\sigma_1(V^{-1})=\bar b$.

If $\rho_3=0$, equivalently $0\in R$,
equation~\eqref{eq:b-quadratic} has the unique root
$b=\sqrt{\rho_2}\in\gf_q$.
The construction above and the uniqueness established after
\eqref{3:quartetp} show that $R$ has exactly one preimage $V$.
Since $V^{-1}$ is also a preimage, we have $V=V^{-1}$.

If $\rho_3\ne0$, then \eqref{eq:res-coeff-3} gives
$\bar b=b+\sqrt{\rho_3}\ne b$, so $\bar b$ is the other root
of \eqref{eq:b-quadratic}.
Moreover, $\sigma_1(V^{-1})=\bar b\ne b=\sigma_1(V)$,
so $V\ne V^{-1}$.
Thus $V$ and $V^{-1}$ are two distinct preimages of $R$.
As there are at most two preimages, we conclude that
$\operatorname{Res}^{-1}(R)=\{V,V^{-1}\}$.
\end{proof}

We now apply Theorem~\ref{resv} to determine conditions defining the set $W$, which in turn can be converted into equations over $\mathbb{F}_q$. 

For $V\in\cQ$, $c=\sigma_2(V) \in \gf_q$, and
Lemma~\ref{lem:determinant} says that $V\in W$ precisely when
\begin{equation}\label{dsigma}
 (b^2+c)(\bar b^2+\bar c)=b\bar b \quad \Longleftrightarrow \quad b\bar b+b^2\bar b^2+c(b^2+\bar b^2)+c^2=0.
\end{equation}
We partition $W=W_0\sqcup W_1$, where
\[
W_0=\{V\in W:0\in\operatorname{Res}(V)\},
\qquad
W_1=\{V\in W:0\notin\operatorname{Res}(V)\}.
\]
By this partition we have
\begin{eqnarray} \label{3:W0W1}
|W|=|W_0|+|W_1|.\end{eqnarray}
We treat the two cases separately in the following subsections.
\subsection{Resolvents containing zero}

Suppose $V\in W_0$, and let
\[
 \operatorname{Res}(V)=\{0,r_1,r_2\},
\]
where $r_1,r_2\in\D_1$ and $r_1\ne r_2$. Lemma~\ref{expressv} gives $b=\sqrt{r_1r_2}$ and
$c=r_1+r_2$. Substituting these values into
\eqref{dsigma}, we obtain
\[
r_1r_2+r_1^2r_2^2+r_1^2+r_2^2=0.
\]
For $\varepsilon\in\{0,1\}$, write
\[\cH_{\varepsilon}=\left\{x\in\gf_q:\operatorname{tr}(x)=\varepsilon\right\}.\]
The indicator of $\cH_{\varepsilon}$ is
\begin{eqnarray} \label{5:h1id} \mathbf 1_{\cH_{\varepsilon}}(x)=(1+(-1)^{\varepsilon}\chi(x))/2,\quad x\in\gf_q.\end{eqnarray}
Set $x=r_1^{-1}$ and $y=r_2^{-1}$.
Then $x,y\in\cH_1$ and $x\ne y$.
Substituting $r_1=x^{-1}$ and $r_2=y^{-1}$
into the preceding equation and simplifying, we obtain $x^2+xy+y^2+1=0$.
Define
\begin{equation}\label{eq:J2-definition}
	J_2=\left\{(x,y)\in\cH_1^2:x^2+xy+y^2+1=0\right\}.
\end{equation}
Using these indicators, we will evaluate $|J_2|$ by
character sums in Section~\ref{coniccounts}.

Setting $x=y$ in the equation gives $x^2+1=0$,
so $x=y=1$. This pair belongs to $\cH_1^2$
exactly when $m$ is odd.  After removing this solution when it occurs,
each remaining unordered pair determines a resolvent
$\{0,x^{-1},y^{-1}\}$.  By Theorem~\ref{resv}, this resolvent
corresponds to a unique quartet in $W_0$.  Dividing by the two
orderings therefore gives
\begin{equation}\label{card1}
 |W_0|=\frac{1}{2} \left(|J_2|-\frac{1-(-1)^m}{2}\right).
\end{equation}

\subsection{Nonzero resolvents}

Suppose $V\in W_1$ and let
\[
 \operatorname{Res}(V)=R=\{r_1,r_2,r_3\}\subset\D_1,\qquad |R|=3.
\]
Put $x=r_1^{-1}$, $y=r_2^{-1}$, $z=r_3^{-1}$.  Then
$x,y,z\in\cH_1$ are pairwise distinct.  On substituting the identities
of Lemma~\ref{expressv} into \eqref{dsigma}, one obtains
\begin{equation}\label{eq:Phi-definition}
\begin{split}
 \Phi(x,y,z)={}&(x^2+y^2+z^2+xy+xz+yz)(xy+xz+yz+1)\\
 &+(x+y+z)(x+y)(x+z)(y+z)=0.
\end{split}
\end{equation}
Define
\begin{equation}\label{eq:J3-definition}
 J_3=\left\{(x,y,z)\in\cH_1^3:\Phi(x,y,z)=0,
             \ x,y,z\text{ pairwise distinct}\right\}.
\end{equation}
Each admissible nonzero resolvent has six orderings and corresponds
to two distinct quartets $V,V^{-1}$ in $W_1$ by
Theorem~\ref{resv}. Hence
\begin{equation}\label{card2}
 |W_1|=\frac{2|J_3|}{6}=\frac{|J_3|}{3}.
\end{equation}
Combining \eqref{card1}, \eqref{card2} and \eqref{3:W0W1} we have 
\begin{equation*}
  |W|=\frac12\left(|J_2|-\frac{1-(-1)^m}{2}\right)+\frac{|J_3|}{3}.
\end{equation*}
Section~\ref{coniccounts} evaluates $|J_2|$ and $|J_3|$.
% ===== End included source: chap3.tex =====

\section{Counting solutions by conic character sums}\label{coniccounts}
% ===== Begin included source: chap5.tex =====
Section~\ref{resolventsection} reduces the normalized quartet
count $|W|$ to the cardinalities of the solution sets
$J_2$ and $J_3$ defined in \eqref{eq:J2-definition}
and \eqref{eq:J3-definition}.
In this section, we evaluate these cardinalities using
character sums over affine conics and combine the results
to express $|W|$ in terms of the mixed Kloosterman sum
$S_m$ defined below.

For $a\in\gf_q$, let
\[
 K(a)=\sum_{x\in\gf_q^*}\chi(x+a/x).
\]
We shall use $K(0)=-1$ and $K(a^2)=K(a)$ from
Lemma~\ref{sumlemma}, together with the mixed sum
\begin{equation}\label{eq:Sm-definition}
 S_m=\sum_{c\in\gf_q^*}\chi(c+c^{-1})K(c).
\end{equation}
Appendix~\ref{Kloostermansum} gives the Kloosterman sum identities needed in this
section and expresses $S_m$ in terms of $T_m$, the $m$th term
of the sequence defined in \eqref{1:tm}.
\subsection{Character sums over an affine conic}

We use the following identity for character sums over an
affine conic to compute $|J_2|$ and $|J_3|$.

\begin{lemma}\label{lem:conic-Fourier}
	Let $c\in\gf_q^*$ and $e,\alpha,\beta\in\gf_q$, and define
	\[
	\mathcal C_{c,e}
	=\left\{(x,y)\in\gf_q^2:x^2+cxy+ey^2+1=0\right\},
	\qquad
	A_{c,e}^{\alpha,\beta}
	=\sum_{(x,y)\in\mathcal C_{c,e}}\chi(\alpha x+\beta y).
	\]
	Then
	\begin{equation}\label{eq:conic-Fourier}
		A_{c,e}^{\alpha,\beta}
		=q\,\mathbf 1_{(\alpha,\beta)=(0,0)}
		+\chi({e}/{c^2})K\!\left(
		\frac{e\alpha^2+c\alpha\beta+\beta^2}{c^2}\right).
	\end{equation}
	Here and below, $\mathbf 1_P$ equals $1$ if the condition
	$P$ holds and $0$ otherwise.
	In particular,
	\[
	|\mathcal C_{c,e}|=A_{c,e}^{0,0}=q-\chi({e}/{c^2}).
	\]
\end{lemma}

\begin{proof}
	By character orthogonality,
	\begin{align*}
		A_{c,e}^{\alpha,\beta}
		&=\sum_{x,y\in\gf_q}\chi(\alpha x+\beta y)
		\mathbf 1_{x^2+cxy+ey^2+1=0}\\
		&=\frac1q\sum_{x,y\in\gf_q}\chi(\alpha x+\beta y)
		\sum_{t\in\gf_q}\chi\bigl(t(x^2+cxy+ey^2+1)\bigr)\\
		&=\frac1q\sum_{t\in\gf_q}\chi(t)
		\sum_{x,y\in\gf_q}
		\chi(tx^2+tcxy+tey^2+\alpha x+\beta y).
	\end{align*}
	The term $t=0$ contributes
	$q\mathbf 1_{(\alpha,\beta)=(0,0)}$.
	For $t\ne0$, additivity and $\chi(z^2)=\chi(z)$ allow us
	to replace $tcxy$, $\alpha x$, and $\beta y$ inside the
	character by their squares. Thus the inner sum is
	\begin{align*}
		&\sum_{x,y\in\gf_q}
		\chi(tx^2+tcxy+tey^2+\alpha x+\beta y)\\
		&\quad=\sum_{y\in\gf_q}\chi\bigl(y^2(te+\beta^2)\bigr)
		\sum_{x\in\gf_q}
		\chi\bigl(x^2(t+t^2c^2y^2+\alpha^2)\bigr)\\
		&\quad=q\sum_{y\in\gf_q}\chi\bigl(y^2(te+\beta^2)\bigr)
		\mathbf 1_{t+t^2c^2y^2+\alpha^2=0},
	\end{align*}
	by character orthogonality.
	As $tc\ne0$, the last condition determines a unique $y$, with 
	\[
	y^2=\frac{t+\alpha^2}{t^2c^2}.
	\]
	Consequently,
	\[
	A_{c,e}^{\alpha,\beta}
	=q\mathbf 1_{(\alpha,\beta)=(0,0)}
	+\sum_{t\in\gf_q^*}
	\chi\!\left(
	t+\frac{(t+\alpha^2)(te+\beta^2)}{t^2c^2}\right).
	\]
	Expanding the fraction and using $\chi\!\left(\frac{\alpha^2\beta^2}{t^2c^2}\right)	=\chi\!\left(\frac{\alpha\beta}{tc}\right)$	gives
	\begin{align*}
		A_{c,e}^{\alpha,\beta}
		&=q\mathbf 1_{(\alpha,\beta)=(0,0)}
		+\sum_{t\in\gf_q^*}
		\chi\!\left(t+\frac{e}{c^2}+
		\frac{e\alpha^2+c\alpha\beta+\beta^2}{c^2t}\right)\\
		&=q\mathbf 1_{(\alpha,\beta)=(0,0)}
		+\chi(\frac{e}{c^2})K\!\left(
		\frac{e\alpha^2+c\alpha\beta+\beta^2}{c^2}\right).
	\end{align*}
	This proves \eqref{eq:conic-Fourier}.
	Setting $\alpha=\beta=0$ and using $K(0)=-1$
	gives the point count.
\end{proof}

\subsection{The cardinality of \texorpdfstring{$J_2$}{J2}}

We have $J_2=\mathcal C_{1,1}\cap\cH_1^2$.
We compute its cardinality by expanding the trace-one
indicators and applying Lemma~\ref{lem:conic-Fourier}.

\begin{lemma}\label{no1}
The set $J_2$ defined in \eqref{eq:J2-definition} has
cardinality
\begin{equation}\label{eq:J2-value}
|J_2|=\frac{q-(-1)^m\left(K(1)+1\right)}4.
\end{equation}
\end{lemma}

\begin{proof}
Using \eqref{5:h1id} gives
\begin{align*}
 4|J_2|
 &=\sum_{(x,y)\in\mathcal C_{1,1}}
   (1-\chi(x))(1-\chi(y))\\
 &=\sum_{(x,y)\in\mathcal C_{1,1}}
   \bigl(1-\chi(x)-\chi(y)+\chi(x+y)\bigr)\\
 &=A_{1,1}^{0,0}-A_{1,1}^{1,0}-A_{1,1}^{0,1}+A_{1,1}^{1,1}.
\end{align*}
Taking $c=e=1$ in Lemma~\ref{lem:conic-Fourier}, we obtain
\[
 A_{1,1}^{0,0}=q-(-1)^m,\qquad
 A_{1,1}^{1,0}=A_{1,1}^{0,1}=A_{1,1}^{1,1}=(-1)^mK(1).
\]
Substitution yields
\[
 4|J_2|=q-(-1)^m\bigl(K(1)+1\bigr),
\]
which proves \eqref{eq:J2-value}.
\end{proof}

\subsection{The cardinality of \texorpdfstring{$J_3$}{J3}}

To compute $|J_3|$, we expand the trace-one indicators and use
permutation symmetry to reduce the calculation to four
character sums over a common set.
We parametrize this set by a family of affine conics with one point
removed from each and evaluate the four sums using
Lemma~\ref{lem:conic-Fourier}.

\begin{lemma}\label{no2}
The set $J_3$ defined in \eqref{eq:J3-definition} has
cardinality
\begin{equation}\label{eq:J3-value}
 |J_3|=\frac{q^2-7q+8+
 (-1)^m\left(S_m+3K(1)-6\right)}8.
\end{equation}
\end{lemma}

\begin{proof}
Recall $\Phi(x,y,z)$ from \eqref{eq:Phi-definition}.
Put $h(x,y,z)=x^2+y^2+z^2+xy+xz+yz$, and consider
the symmetric set
\begin{equation*}
	\Omega=\left\{(x,y,z)\in\gf_q^3:
	x,y,z\text{ pairwise distinct},\
	h(x,y,z)\ne0,\ \Phi(x,y,z)=0\right\}.
\end{equation*}
To show that $J_3\subset\Omega$, only the condition $h(x,y,z)\ne0$
needs checking.  If $(x,y,z)\in J_3$ and $h(x,y,z)=0$, then
\eqref{eq:Phi-definition} gives
\[
 (x+y+z)(x+y)(x+z)(y+z)=0.
\]
Pairwise distinctness forces $x+y+z=0$, contradicting
$\operatorname{tr}(x+y+z)=1$ since $x,y,z\in\cH_1$.  Thus $ J_3=\Omega\cap\cH_1^3$. Define
\[
 A_0=|\Omega|,\quad
 A_1=\sum_{(x,y,z)\in\Omega}\chi(x),\quad
 A_2=\sum_{(x,y,z)\in\Omega}\chi(x+y),\quad
 A_3=\sum_{(x,y,z)\in\Omega}\chi(x+y+z).
\]
Expanding the trace-one indicators \eqref{5:h1id} and using permutation symmetry gives
\begin{equation}\label{eq:J3-symmetry}
\begin{aligned}
 8|J_3|
 &=\sum_{(x,y,z)\in\Omega}(1-\chi(x))(1-\chi(y))(1-\chi(z))\\
 &=\sum_{(x,y,z)\in\Omega}\Bigl(1-\chi(x)-\chi(y)-\chi(z)\\
 &\qquad\qquad{}+\chi(x+y)+\chi(x+z)+\chi(y+z)-\chi(x+y+z)\Bigr)\\
 &=A_0-3A_1+3A_2-A_3.
\end{aligned}
\end{equation}

For $(x,y,z)\in\Omega$, the condition $y\ne z$ allows us to put
\[
 s=x+y+z,\qquad u=y+z,\qquad
 \theta=\frac{x+z}{y+z}.
\]
These coordinates recover the triple by
\[
 (x,y,z)=\bigl(s+u,s+\theta u,s+(1+\theta)u\bigr).
\]
The recovered coordinates are pairwise distinct exactly when
$u\ne0$ and $\theta\ne0,1$.
With $E_\theta:=\theta^2+\theta+1$, substitution gives
\begin{gather*}
 h(x,y,z)=E_\theta u^2,\qquad
 xy+xz+yz=s^2+E_\theta u^2,\\
 (x+y)(x+z)(y+z)=(E_\theta+1)u^3.
\end{gather*}
Hence the additional condition $h(x,y,z)\ne0$ is equivalent to
$E_\theta\ne0$.  Set
\[
 \Theta:=\left\{\theta\in\gf_q:\theta\ne0,1,\ E_\theta\ne0\right\},
 \qquad C_\theta:=\frac{E_\theta+1}{E_\theta}\quad(\theta\in\Theta).
\]
In particular, $C_\theta\ne0$ because
$E_\theta+1=\theta(\theta+1)$.  Equation~\eqref{eq:Phi-definition}
now becomes
\[
 \Phi(x,y,z)
 =E_\theta u^2\bigl(s^2+C_\theta su+E_\theta u^2+1\bigr).
\]
Since $E_\theta u^2\ne0$, the condition $\Phi(x,y,z)=0$ is equivalent
to $s^2+C_\theta su+E_\theta u^2+1=0$.  Thus the precise parameter set is
\begin{equation*}
\begin{aligned}
 \widetilde\Omega
 &=\left\{(s,u,\theta)\in\gf_q\times\gf_q^*\times\Theta:
 (s,u)\in\mathcal C_{C_\theta,E_\theta}\right\}.
\end{aligned}
\end{equation*}
Conversely, for every $(s,u,\theta)\in\widetilde\Omega$, the recovered
triple is pairwise distinct and satisfies $h(x,y,z)\ne0$ and
$\Phi(x,y,z)=0$.  Thus the coordinate change is a bijection between
$\Omega$ and $\widetilde\Omega$.  In particular,
\begin{equation}\label{eq:centered-bijection}
 \Omega=\left\{\bigl(s+u,s+\theta u,s+(1+\theta)u\bigr):
 (s,u,\theta)\in\widetilde\Omega\right\}.
\end{equation}
For each $\theta\in\Theta$, the only conic point with $u=0$ is $(1,0)$.
Using \eqref{eq:centered-bijection}, we therefore obtain
\begin{align*}
 A_0
 &=\sum_{\theta\in\Theta}
   \sum_{\substack{(s,u)\in\mathcal C_{C_\theta,E_\theta}\\u\ne0}}1
 =\sum_{\theta\in\Theta}\left(A_{C_\theta,E_\theta}^{0,0}-1\right),\\
 A_1
 &=\sum_{\theta\in\Theta}
   \sum_{\substack{(s,u)\in\mathcal C_{C_\theta,E_\theta}\\u\ne0}}\chi(s+u)
 =\sum_{\theta\in\Theta}\left(A_{C_\theta,E_\theta}^{1,1}-(-1)^m\right),\\
 A_2
 &=\sum_{\theta\in\Theta}
   \sum_{\substack{(s,u)\in\mathcal C_{C_\theta,E_\theta}\\u\ne0}}\chi((1+\theta)u)
 =\sum_{\theta\in\Theta}\left(A_{C_\theta,E_\theta}^{0,1+\theta}-1\right),\\
 A_3
 &=\sum_{\theta\in\Theta}
   \sum_{\substack{(s,u)\in\mathcal C_{C_\theta,E_\theta}\\u\ne0}}\chi(s)
 =\sum_{\theta\in\Theta}\left(A_{C_\theta,E_\theta}^{1,0}-(-1)^m\right).
\end{align*}
Here $x=s+u$, $x+y=(1+\theta)u$, and $x+y+z=s$; the subtracted terms
are the values of the corresponding summands at $(1,0)$.

We now apply Lemma~\ref{lem:conic-Fourier}, which is applicable since
$C_\theta\ne0$ for every $\theta\in\Theta$.
We first compute $\chi(E_\theta/C_\theta^2)$.  We have
\[
 \frac{E_\theta}{C_\theta^2}
 =\frac{E_\theta^3}{(E_\theta+1)^2}
 =E_\theta+\frac1{E_\theta+1}+\frac1{(E_\theta+1)^2}.
\]
The last two terms have equal trace, while
$\operatorname{tr}(E_\theta)=\operatorname{tr}(1)$, so $\chi(E_\theta/C_\theta^2)=(-1)^m$. For $A_1,A_2,A_3$, the arguments of $K$ simplify as follows:
\[
 \frac{E_\theta+C_\theta+1}{C_\theta^2}=E_\theta,\qquad
 \frac{(1+\theta)^2}{C_\theta^2}=(E_\theta/\theta)^2,\qquad
 \frac{E_\theta}{C_\theta^2}=\frac{E_\theta^3}{(E_\theta+1)^2}.
\]
Substituting these expressions into the four sums and using
$K(a^2)=K(a)$ gives
\begin{equation}\label{eq:conic-contribution-table}
\begin{aligned}
 A_0&=|\Theta|(q-1-(-1)^m),\\
 A_1&=(-1)^m\left(\sum_{\theta\in\Theta}K(E_\theta)-|\Theta|\right),\\
 A_2&=(-1)^m\sum_{\theta\in\Theta}K\left(\frac{E_\theta}{\theta}\right)-|\Theta|,\\
 A_3&=(-1)^m\left(
   \sum_{\theta\in\Theta}K\!\left(\frac{E_\theta^3}{(E_\theta+1)^2}\right)-|\Theta|
   \right).
\end{aligned}
\end{equation}

It remains to evaluate the three sums over $\Theta$.
The identities in Appendix~\ref{app:conic-identities} sum over all
$\theta\ne0,1$, so we must remove the roots of $E_\theta=0$.
By Lemma~\ref{rootinfq}, the equation $E_\theta=0$, or
$\theta^2+\theta+1=0$, has $1+(-1)^m$ roots.  Consequently, $|\Theta|=q-3-(-1)^m$.

At each excluded root, all three Kloosterman arguments in
\eqref{eq:conic-contribution-table} are zero.  Since $K(0)=-1$,
removing these roots adds $1+(-1)^m$ to each sum.
Equations~\eqref{eq:ksum1}--\eqref{eq:ksum3}, with
$E_\theta+1=\theta^2+\theta$, therefore give
\begin{align*}
 \sum_{\theta\in\Theta}K(E_\theta)
 &=q-2K(1)+1+(-1)^m,\\
 \sum_{\theta\in\Theta}K\left(\frac{E_\theta}{\theta}\right)
 &=S_m-K(1)+1+(-1)^m,\\
 \sum_{\theta\in\Theta}K\!\left(\frac{E_\theta^3}{(E_\theta+1)^2}\right)
 &=2S_m-q+1+(-1)^m.
\end{align*}
Substituting these evaluations into \eqref{eq:conic-contribution-table}
and then using \eqref{eq:J3-symmetry}, we obtain
\begin{align*}
 8|J_3|
 &=|\Theta|\bigl(q-4+3(-1)^m\bigr)\\
 &\quad+(-1)^m\bigl(S_m+3K(1)-2q-1-(-1)^m\bigr)\\
 &=q^2-7q+8+(-1)^m\bigl(S_m+3K(1)-6\bigr),
\end{align*}
where the last step uses $|\Theta|=q-3-(-1)^m$.
This proves \eqref{eq:J3-value}.
\end{proof}

\subsection{The normalized quartet count}

We now combine the two cardinalities using the multiplicities
established in Section~\ref{resolventsection}.  The terms involving
$K(1)$ cancel.

\begin{theorem}\label{thm:Wcount}
The normalized quartet set $W$ defined in \eqref{3:ZW1} has
cardinality
\begin{equation}\label{eq:W-count}
|W|=\frac{q^2-4q+2+(-1)^m(S_m-3)}{24}.
\end{equation}
\end{theorem}

\begin{proof}
Substituting \eqref{eq:J2-value} into \eqref{card1} and
\eqref{eq:J3-value} into \eqref{card2} gives, respectively,
\begin{align*}
 |W_0|&=\frac{q-2+(-1)^m-(-1)^mK(1)}8,\\
 |W_1|&=\frac{q^2-7q+8+
 (-1)^m\bigl(S_m+3K(1)-6\bigr)}{24}.
\end{align*}
Adding these expressions and using $|W|=|W_0|+|W_1|$ proves
\eqref{eq:W-count}.
\end{proof} 
% ===== End included source: chap5.tex =====

\section{Completion of the cross-correlation distribution}\label{main}
% ===== Begin included source: chap6.tex =====
Section~\ref{coniccounts} gives the normalized quartet count $|W|$
in terms of $S_m$.  Combining this count with the incidence identity
\eqref{eq:W-incidence} and the evaluation of $S_m$ in
Appendix~\ref{app:Sm-evaluation} determines $N_5$.

\begin{theorem}\label{N5}
The number $N_5$ of elements $a\in\gf_{q^2}$ for which the
polynomial $f_a$ in \eqref{fa} has exactly five distinct roots on
$U_{q+1}$ is
\begin{equation}\label{eq:N5}
 N_5=\frac{q^2-7q+10-
 (-1)^m(T_m-5q+11)}{120}.
\end{equation}
\end{theorem}

\begin{proof}
	Using the incidence identity \eqref{eq:W-incidence},
	the quartet count \eqref{eq:W-count}, and the value of
	$N_4$ obtained in Section~\ref{rootfreq}, we have
	
\[
 N_5=\frac{q^2-8q+10+(-1)^m(S_m+4q-11)}{120}.
\]
By Theorem~\ref{smtm} in Appendix~\ref{app:Sm-evaluation}, $S_m=-T_m+\bigl(1+(-1)^m\bigr)q$. 
Substituting this expression and simplifying proves \eqref{eq:N5}.
\end{proof}

Substituting \eqref{eq:N5} into \eqref{eq:frequencies-in-N5}
gives the remaining root frequencies $N_0,N_1,N_2,N_3$. As explained in Section~\ref{rootfreq}, we obtain the correlation
frequencies from the counts $N_i$ by excluding $a=0$.  Since
$f_0(X)=X^7+1$ and $\gcd(7,q+1)=1$, its only root on $U_{q+1}$ is
$1$, so $N(0)=1$.  Thus only the frequency at $-1$ changes, from
$N_1$ to
\[
 N_1-1=\frac{9q^2+5q-22+(-1)^m(3T_m-11q+1)}{24}.
\]
All other frequencies remain $N_i$.  The resulting frequencies sum
to $q^2-1$ and give precisely the distribution in
Theorem~\ref{distribution}.

\begin{rmk}
	Let $B_5$ be the number of weight-five codewords in the binary
	Zetterberg code of length $q+1$. Xia et al.~\cite{XLZH16}
	used this count to determine the cross-correlation distribution
	for $s=3$. In the present $s=4$ problem, the same count is
	related to $N_5$ when $m$ is odd. Indeed, expressed in terms of $T_m$, the formula in
	\cite[Proposition~6.1 and Table~6.2]{SV91} gives,
	for odd $m\geq3$,
	\[
	\frac{B_5}{q+1}
	=\frac{q^2-12q+21+T_m}{120}=N_5.
	\]
	The equality also holds for $m=1$, when both counts are zero.
\end{rmk}

% ===== End included source: chap6.tex =====

\section{Concluding remarks}\label{conclusion}
We have determined the complete cross-correlation distribution
for $d=4(2^m-1)+1$ for every positive integer $m$.
The proof uses three steps: the normalization
of four-root incidences, reconstruction of a product-one
quartet from its cubic resolvent, and the
reduction of the computation of $|J_2|$ and $|J_3|$ to
character sums over $\mathbb{F}_q$.
These steps may also be useful for other Niho-type
exponents whose unit-circle root equations have few nonzero terms.

\appendix
\section{The determinant identity}\label{app:determinant}
% ===== Begin included source: chap7.tex =====
We prove Lemma~\ref{lem:determinant} by reducing the columns of the
evaluation matrix using the quartic with roots $v_1,\ldots,v_4$.

\begin{proof}
For $j\geq0$, write $C_j=(v_1^j,v_2^j,v_3^j,v_4^j)^{\mathsf T}$.
Thus $D(V)=\det\left[C_7,C_4,C_3,C_0\right]$. Put
\[
 P(X)=\prod_{i=1}^4(X+v_i)=X^4+bX^3+cX^2+hX+p,
\]
where $b=\sigma_1(V)$, $c=\sigma_2(V)$, $h=\sigma_3(V)$, and $p=\sigma_4(V)$. Since $P(v_i)=0$ for $i=1,2,3,4$, we have
\begin{equation}\label{3:col1}
 C_4=bC_3+cC_2+hC_1+pC_0.
\end{equation}
We also have the identity
\[
\begin{aligned}
 X^7+b^3X^4+(b^2c+c^2+p)X^3+(b^2h+bp)X^2+\bigl(h^2+p(b^2+c)\bigr)X+hp\\
  =P(X)\bigl(X^3+bX^2+(b^2+c)X+h\bigr).\\
\end{aligned}
\]
Evaluating at $X=v_i$ and using $P(v_i)=0$ gives
\[
 v_i^7=b^3v_i^4+(b^2c+c^2+p)v_i^3+Av_i^2+Bv_i+hp, 
\]
where we set $ A=b^2h+bp$ and $B=h^2+p(b^2+c)$. Equivalently,
\begin{equation}\label{3:col2}
 C_7=b^3C_4+(b^2c+c^2+p)C_3+AC_2+BC_1+hpC_0.
\end{equation}
By \eqref{3:col1} and \eqref{3:col2}, we have
\begin{align*}
 D(V)
 &=\det[AC_2+BC_1,cC_2+hC_1,C_3,C_0]\\
 &=Ah\det[C_2,C_1,C_3,C_0]+Bc\det[C_1,C_2,C_3,C_0]\\
 &=(Ah+Bc)\det[C_0,C_1,C_2,C_3].
\end{align*}
The Vandermonde determinant gives
\[
\det[C_0,C_1,C_2,C_3]=\prod_{1\leq i<j\leq4}(v_i+v_j)=\Delta(V).
\]
Using $c=p\bar c$ and $h=p\bar b$ from \eqref{3:bcp}, we obtain
\[
\begin{aligned}
 Ah+Bc&=(b^2h+bp)h+\bigl(h^2+p(b^2+c)\bigr)c\\
 &=p^2\bigl((b^2+c)(\bar b^2+\bar c)+b\bar b\bigr).
\end{aligned}
\]
This proves \eqref{eq:correct-determinant}.
Since $\Delta(V)$ and $p$
are nonzero, the stated criterion for $D(V)=0$ follows.
\end{proof}
% ===== End included source: chap7.tex =====

\section{Kloosterman identities and the mixed sum \texorpdfstring{$S_m$}{Sm}}\label{Kloostermansum}
% ===== Begin included source: chap4.tex =====
This appendix supplies the character-sum identities used in
Sections~\ref{coniccounts} and~\ref{main}.  We first record basic
Kloosterman identities, then evaluate $S_m$ through Carlitz's
recurrence, and finally prove the three summation identities used
in the conic counts.

\subsection{Basic identities}

For a finite extension $E$ of $\gf_2$ and $a\in E$, define the
Kloosterman sum
\[
 K_E(a)=\sum_{x\in E^*}(-1)^{\tr_{E/\gf_2}(x+a/x)}.
\]
When $E=\gf_q$ we write $K(a)=K_{\gf_q}(a)$.  Notice in particular that
$K(0)=-1$. We prove the following identities.

\begin{lemma}\label{sumlemma}
For $a\in \gf_q^*$,
\[
 K(a^2)=K(a),\qquad
 \sum_{x\in \gf_q^*}\chi(ax+a/x)=K(a),\qquad
 \sum_{c\in \gf_q}K(c)=0.
\]
\end{lemma}

\begin{proof}
The first identity follows by replacing $x$ with $y^2$ and using trace
invariance under squaring. For the second identity, the substitution $y=ax$ transforms
the sum into $K(a^2)$, which equals $K(a)$ by the first identity.  Finally,
\[
 \sum_{c\in \gf_q}K(c)
 =\sum_{x\in \gf_q^*}\chi(x)\sum_{c\in \gf_q}\chi(c/x)=0
\]
by additive-character orthogonality.
\end{proof}

\subsection{Evaluation of the mixed sum}\label{app:Sm-evaluation}

We evaluate the mixed sum $S_m$ defined in \eqref{eq:Sm-definition}.
We first prove the following identity for $S_m$.

\begin{lemma}\label{sm1}
\begin{equation*}
 \sum_{\substack{t\in \gf_q\\t^2+t+1\ne0}}
 K\!\left(\frac{t}{t^2+t+1}\right)=1+(-1)^m S_m.
\end{equation*}
\end{lemma}

\begin{proof}
Put $c=t/(t^2+t+1)$. For $c=0$ and $c=1$, the unique solutions are
$t=0$ and $t=1$, respectively.  For $c\ne0,1$, the equation $c=t/(t^2+t+1)$ is equivalent to
\[
t^2+(1+c^{-1})t+1=0,
\]
and therefore has $1+\chi(c/(c+1))$ solutions by
Lemma~\ref{rootinfq} and trace invariance under squaring.  Therefore,
\begin{align*}
	&\sum_{\substack{t\in\gf_q\\t^2+t+1\ne0}}
	K\!\left(\frac{t}{t^2+t+1}\right)\\
	&\quad=K(0)+K(1)
	+\sum_{c\ne0,1}
	\left(1+\chi\!\left(\frac{c}{c+1}\right)\right)K(c)\\
	&\quad=\sum_{c\ne0,1}
	\chi\!\left(\frac{c}{c+1}\right)K(c)\\
	&\quad=(-1)^m\sum_{c\ne0,1}\chi(c^{-1})K(c+1),
\end{align*}
where the second equality uses $\sum_{c\in\gf_q}K(c)=0$,
and the last follows from the substitution $c\mapsto c+1$
and $\chi(1)=(-1)^m$.

Since $K(0)=-1$ and $\chi(1)=(-1)^m$, we have
\begin{align*}
 (-1)^m\sum_{c\ne0,1}\chi(c^{-1})K(c+1)
 &=1+(-1)^m\sum_{c\in\gf_q^*}\chi(c^{-1})K(c+1)\\
 &=1+(-1)^m\sum_{c,x\in\gf_q^*}
 \chi\!\left(c^{-1}+x+\frac{c+1}{x}\right).
\end{align*}
Substituting $c=ux$ in the double sum gives
\[
\begin{aligned}
 &\sum_{c,x\in\gf_q^*}
   \chi\!\left(c^{-1}+x+\frac{c+1}{x}\right)\\
 &\quad=\sum_{x\in\gf_q^*}\chi(x+x^{-1})
   \sum_{u\in\gf_q^*}\chi\!\left(u+\frac1{xu}\right)\\
 &\quad=\sum_{x\in\gf_q^*}\chi(x+x^{-1})K(x^{-1})=S_m,
\end{aligned}
\]
where the last equality follows from the substitution $x\mapsto x^{-1}$.
\end{proof}

We evaluate $S_m$ separately according to whether $m$ is odd or even.

\begin{lemma}\label{lem:Sm-split-nonsplit}
Let $\omega\in\gf_4\setminus\gf_2$ satisfy
$\omega^2+\omega+1=0$.  Then
\begin{equation}\label{eq:Sm-parity}
 S_m=\begin{cases}
 -K_{\gf_{q^2}}(\omega),&m\text{ odd},\\
 K_{\gf_q}(\omega)^2,&m\text{ even}.
 \end{cases}
\end{equation}
\end{lemma}

\begin{proof}
For $a,b\in \gf_q$ put $Q(a,b)=a^2+ab+b^2$ and consider
\[
 \mathcal A=\sum_{\substack{a,b\in \gf_q\\Q(a,b)\ne0}}
 \chi\!\left(b+\frac{a}{Q(a,b)}\right).
\]
The terms with $b=0$ sum to $-1$.  For $b\ne0$, write $a=tb$;
then
\[
 Q(a,b)=b^2(t^2+t+1),\qquad
 b+\frac{a}{Q(a,b)}
 =b+\frac{t}{(t^2+t+1)b}.
\]
Lemma~\ref{sm1} therefore gives
\begin{equation}\label{eq:Sigma-Sm}
 \mathcal A=-1+
 \sum_{\substack{t\in\gf_q\\t^2+t+1\ne0}}K\!\left(\frac{t}{t^2+t+1}\right)
 =(-1)^m S_m.
\end{equation}

Suppose first that $m$ is odd.  Then $\omega\notin \gf_q$,
$\omega^q=\omega^2$, and every $z\in\gf_{q^2}$ can be written uniquely as
$z=a+b\omega$ with $a,b\in\gf_q$.
Since $Q(a,b)=z\bar z$, we have $Q(a,b)\ne0$ if and only if $z\ne0$.
For $z\ne0$,
\[
\begin{aligned}
 \tr_{\gf_{q^2}/\gf_q}\!\left(z+\frac{\omega}{z}\right)
 &=z+\bar z+\frac{\omega\bar z+\omega^2z}{z\bar z}\\
 &=b+\frac{a}{Q(a,b)}.
\end{aligned}
\]
Trace transitivity shows that $\mathcal A=K_{\gf_{q^2}}(\omega)$.  Since
$m$ is odd, \eqref{eq:Sigma-Sm} gives the first line of
\eqref{eq:Sm-parity}.

If $m$ is even, $\omega,\omega^2\in\gf_q$.
The invertible substitution $x=a+b\omega$, $y=a+b\omega^2$
has inverse $b=x+y$, $a=\omega^2x+\omega y$, and gives $Q(a,b)=xy$.
Thus $Q(a,b)\ne0$ if and only if $x,y\in\gf_q^*$.  For these values,
\[
b+\frac{a}{Q(a,b)}
=x+y+\frac{\omega}{x}+\frac{\omega^2}{y}.
\]
Hence
\[
 \mathcal A=K_{\gf_q}(\omega)K_{\gf_q}(\omega^2)=K_{\gf_q}(\omega)^2,
\]
where Lemma~\ref{sumlemma} was used in the last equality. Since
$m$ is even,  the second line follows.
\end{proof}

Carlitz's evaluation over finite extensions of $\gf_4$
\cite[Section~7, Eqs.~(7.3)--(7.5)]{C69} supplies the recurrence
that makes \eqref{eq:Sm-parity} explicit. The sum $S^{(m)}(\omega)$ in \cite{C69}
equals $K_{\gf_{4^m}}(\omega)$ after the substitution
$x\mapsto x^{-1}$.

\begin{lemma}[Carlitz]\label{tm}
The sequence $(T_k)_{k\ge0}$ defined in \eqref{1:tm} takes integer values
and satisfies
\[
T_0=-2,\qquad T_1=-1,\qquad
T_{k+2}-T_{k+1}+4T_k=0\quad(k\ge0).
\]
Let $\lambda_1=\frac{1+\sqrt{-15}}2$ and $\lambda_2=\frac{1-\sqrt{-15}}2$. Then
\[
 T_k=-\lambda_1^k-\lambda_2^k
 =-\frac1{2^{k-1}}
   \sum_{r=0}^{\lfloor k/2\rfloor}(-1)^r\binom{k}{2r}15^r.
\]
For $k\ge1$, we have $T_k=K_{\gf_{4^k}}(\omega)$.
\end{lemma}

For odd $m$, the first line of \eqref{eq:Sm-parity} is $S_m=-T_m$.
For even $m$, $K_{\gf_q}(\omega)=T_{m/2}$, and
\[
 S_m=T_{m/2}^2
 =\lambda_1^m+\lambda_2^m+2\cdot4^{m/2}=-T_m+2q.
\]
We have proved the evaluation used in Section~\ref{main}.

\begin{theorem}\label{smtm}
	The mixed Kloosterman sum $S_m$ satisfies
	\begin{equation*}
		S_m=-T_m+(1+(-1)^m)q.
	\end{equation*}
\end{theorem}

\subsection{Summation identities for the conic counts}\label{app:conic-identities}

We prove the three identities used in Section~\ref{coniccounts}. For $\theta\in\gf_q$,
put $E_\theta=\theta^2+\theta+1$. We have the following lemma.

\begin{lemma}\label{ksum1}
\begin{equation}\label{eq:ksum1}
 \sum_{\theta\ne0,1}K(E_\theta)=q-2K(1).
\end{equation}
\end{lemma}

\begin{proof}
For each $u\in\gf_q$, the equation $\theta^2+\theta=u$
has $1+\chi(u)$ solutions by Lemma~\ref{rootinfq}.
Hence, expanding the Kloosterman sum and interchanging
the summations, we obtain
\begin{align*}
 \sum_{\theta\in\gf_q}K(E_\theta)
 &=\sum_{u\in\gf_q}(1+\chi(u))K(u+1)\\
 &=\sum_{u\in\gf_q}K(u+1)
   +\sum_{u\in\gf_q}\chi(u)K(u+1)\\
 &=\sum_{u\in\gf_q}\chi(u)K(u+1)\\
 &=\sum_{x\in\gf_q^*}\chi(x+x^{-1})
   \sum_{u\in\gf_q}\chi\bigl((1+x^{-1})u\bigr)\\
 &=q.
\end{align*}
The third equality uses $\sum_{c\in\gf_q}K(c)=0$.
By character orthogonality, the inner sum in the fourth line is $q$
when $x=1$ and zero otherwise, giving the last equality.
Removing the terms $\theta=0,1$, each of which equals $K(1)$,
proves \eqref{eq:ksum1}.
\end{proof}

\begin{lemma}\label{ksum2}
	With $E_\theta=\theta^2+\theta+1$, we have
	\begin{equation}\label{eq:ksum2}
		\sum_{\theta\ne0,1}K\left(\frac{E_\theta}{\theta}\right)=S_m-K(1).
	\end{equation}
\end{lemma}

\begin{proof}
Expanding the Kloosterman sum and interchanging the two summations gives
\begin{align*}
 \sum_{\theta\ne0,1}K\left(\frac{E_\theta}{\theta}\right)
 &=\sum_{x\in \gf_q^*}\chi(x+x^{-1})
   \sum_{\theta\ne0,1}
   \chi\!\left(\frac{\theta+\theta^{-1}}x\right)\\
 &=\sum_{x\in \gf_q^*}\chi(x+x^{-1})(K(x^{-2})-1)\\
 &=\sum_{x\in \gf_q^*}\chi(x+x^{-1})(K(x^{-1})-1)\\
 &=S_m-K(1).
\end{align*}
In the inner sum, the substitution $y=\theta/x$ gives $K(x^{-2})-1$;
Lemma~\ref{sumlemma} then gives $K(x^{-2})=K(x^{-1})$.
The last equality uses the change of variable $x\mapsto x^{-1}$
and the definition of $K(1)$.
\end{proof}

To prove the third identity, we use the following permutation
property of the trace-one set, due to Hollmann and Xiang
\cite[Theorem~4.1(i)]{HX04}.
We include a short proof for completeness.

\begin{lemma}[\cite{HX04}]\label{lem:trace-one-permutation}
The rational function $ h(s)=s+s^{-1}+s^{-2}$ permutes $\cH_1$.
\end{lemma}

\begin{proof}
For $s\in\cH_1$, trace invariance under squaring gives
$\operatorname{tr}(h(s))=\operatorname{tr}(s)=1$.
Suppose $h(a)=h(b)$ for distinct $a,b\in\cH_1$.
Clearing denominators yields
\[
 (a+b)(a^2b^2+ab+a+b)=0.
\]
Since $a\ne b$, we obtain
\[
 a^2b^2+(a+1)b+a=0.
\]
Since $a\in\cH_1$, we have $a\ne0$. Dividing by $a^2$ gives
\[
b^2+\frac{a+1}{a^2}b+\frac1a=0.
\]
If $a=1$, this equation gives $b=1$, contradicting $a\ne b$.
Thus the linear coefficient is nonzero.
Since this quadratic equation has the root $b\in\gf_q$,
Lemma~\ref{rootinfq} gives
\[
 \operatorname{tr}\!\left(\frac{a^3}{(a+1)^2}\right)=0.
\]
However, trace invariance under squaring gives
\[\operatorname{tr}\!\left(\frac{a^3}{(a+1)^2}\right)=\operatorname{tr}\!\left(a+\frac1{a+1}+\frac1{(a+1)^2}\right)=\operatorname{tr}(a)=1,\]
a contradiction. Thus $h$ permutes $\cH_1$.
\end{proof}

\begin{lemma}\label{ksum3}
	With $E_\theta=\theta^2+\theta+1$, we have
	\begin{equation}\label{eq:ksum3}
		\sum_{\theta\ne0,1}
		K\!\left(\frac{E_\theta^3}{(E_\theta+1)^2}\right)
		=2S_m-q.
	\end{equation}
\end{lemma}

\begin{proof}
The definition of $S_m$ gives
\[
 S_m=\sum_{x,u\in \gf_q^*}\chi(x+x^{-1}+u+x/u).
\]
For fixed $x$, substitute $u=x^2s^2$ and use trace invariance under
squaring.  Then
\[
 S_m=\sum_{s\in \gf_q^*}\sum_{x\in \gf_q^*}
 \chi\bigl(x(1+s)+x^{-1}(1+s^{-2})\bigr).
\]
For $s\ne1$, the substitution $y=(1+s)x$ gives
\[
 \sum_{x\in\gf_q^*}
 \chi\bigl(x(1+s)+x^{-1}(1+s^{-2})\bigr)
 =K\bigl((1+s)(1+s^{-2})\bigr).
\]
For $s\in\gf_q^*$, put
\[
 G(s)=(1+s)(1+s^{-2}).
\]
At $s=1$ the inner sum is $q-1$,
whereas $K(G(1))=K(0)=-1$.  Therefore
\begin{equation}\label{eq:G-unrestricted}
 \sum_{s\in \gf_q^*}K(G(s))=S_m-q.
\end{equation}
Since $G(s)=1+h(s)$ and $h$ permutes $\cH_1$
by Lemma~\ref{lem:trace-one-permutation}, we have
\begin{align*}
 \sum_{s\in\cH_1}K(G(s))
 &=\sum_{s\in\cH_1}K(1+s)\\
 &=\frac12\sum_{s\in \gf_q}(1-\chi(s))K(1+s)=-\frac q2.
\end{align*}
Here $\sum_{s\in\gf_q}K(1+s)=0$. Expanding $K(1+s)$ and
interchanging the summations gives
\begin{align*}
 \sum_{s\in\gf_q}\chi(s)K(1+s)
 &=\sum_{s\in\gf_q}\chi(s)
   \sum_{x\in\gf_q^*}\chi\!\left(x+\frac{1+s}{x}\right)\\
 &=\sum_{x\in\gf_q^*}\chi(x+x^{-1})
   \sum_{s\in\gf_q}\chi\bigl((1+x^{-1})s\bigr)\\
 &=q.
\end{align*}
The last equality follows from character orthogonality. The inner sum
is $q$ when $x=1$ and zero otherwise.
Subtracting this trace-one part from
\eqref{eq:G-unrestricted} gives
\[
 \sum_{s\in\cH_0\setminus\{0\}}K(G(s))=S_m-\frac q2.
\]
Finally, $\theta\mapsto E_\theta+1=\theta^2+\theta$
is two-to-one from $\gf_q\setminus\{0,1\}$
onto $\cH_0\setminus\{0\}$, and
$E_\theta^3/(E_\theta+1)^2=G(E_\theta+1)$.
This proves \eqref{eq:ksum3}.
\end{proof}
% ===== End included source: chap4.tex =====

\section*{Acknowledgments}
The authors acknowledge the use of ChatGPT for exploratory discussions
during the early stages of this research. All mathematical results
were developed, proved, and independently verified by the authors,
who take full responsibility for the content and integrity of this work.
The authors also thank Yaoran Yang and Yutong Zhang for their assistance
with AI tools during the early stages of this research.

M. Xiong's research was supported by the Research Grants Council (RGC)
of Hong Kong under grant no. 16307524.
H. Yan's research was supported by the National Natural Science
Foundation of China under grant no. 12471494 and the National
Cryptologic Science Fund of China under grant no. 2025NCSF02045.

% ===== Begin included source: ref1.tex =====
% SIAM-formatted references. Edit bibliography entries in this file.
% Included directly by Niho-XY-1.tex; BibTeX is not required.
% references.bib is retained only as a backup of the bibliographic data.

% ===== End included source: ref1.tex =====
\end{document}